\documentclass[a4paper,12pt]{article}

\usepackage{docmute}
\usepackage[margin=1.01in]{geometry}
\usepackage{hyperref}
\usepackage[sort,numbers]{natbib}
\usepackage{multirow}
\usepackage{esvect}
\usepackage{amsmath,amsfonts,amssymb,xspace}
\usepackage{amsthm}
\usepackage{tikz}
\usetikzlibrary{fit,shapes}
\usepackage{etoolbox}
\usepackage{datetime}
\usepackage{amssymb}
\newcommand{\edges}[1][]%
{%\end{scope}\end{pgfonlayer}\begin{pgfonlayer}{foreground}\begin{scope}[#1]
}
\theoremstyle{plain}
\newtheorem{theorem}{Theorem}
\newtheorem{lemma}[theorem]{Lemma}
\newtheorem{proposition}[theorem]{Proposition}
\newtheorem{corollary}[theorem]{Corollary}

\theoremstyle{definition} 
\newtheorem{definition}[theorem]{Definition} 
\newtheorem{example}[theorem]{Example} 
 
\newtheorem*{conj*}{Conjecture}

\makeatletter
\def\calign@preamble{%
   &\hfil\strut@
    \setboxz@h{\@lign$\m@th\displaystyle{##}$}%
    \ifmeasuring@\savefieldlength@\fi
    \set@field
    \hfil
    \tabskip\alignsep@
}
\let\cmeasure@\measure@
\patchcmd\cmeasure@{\divide\@tempcntb\tw@}{}{}{}
\patchcmd\cmeasure@{\divide\@tempcntb\tw@}{}{}{}
\patchcmd\cmeasure@{\ifodd\maxfields@
  \global\advance\maxfields@\@ne
  \fi}{}{}{}    
  
\makeatother

\newcommand\tinymatrix[1]
{\left( \hspace{-2pt} \renewcommand\thickspace{\kern2pt} \scriptstyle\begin{smallmatrix} #1 \end{smallmatrix} \hspace{-2pt} \right)}

\newcommand\ignore[1]{}

\DeclareMathOperator{\Tr}{Tr}

\renewcommand\dag{\ensuremath{\dagger}}

\newcommand\C{\ensuremath{\mathbb{C}}}

\newcommand\pdag{{\phantom{\dagger}}}

\usepackage{color}

\def\arraystretch{1.0}
\newcommand\grid[1]{\ensuremath{\def\arraystretch{1.4}\begin{array}{|c|c|c|c|c|c|c|}\hline#1\\\hline\end{array}}}
\newcommand\diag{\mathrm{diag}}
\newcommand\inv{{-1}}
\newcommand\I{\ensuremath{\mathbb I}}
\newcommand\M{\ensuremath{\mathcal M}}
\newcommand\F{\ensuremath{\mathcal F}}
\newcommand\super[2]{\stackrel{\makebox[0pt]{\smash{\tiny #1}}}{#2}}
\newcommand\T{\ensuremath{\mathrm T}}
\newcommand\bra[1]{\langle #1|}
\newcommand\ket[1]{{|} #1 \rangle}
\newcommand\braket[2]{\langle #1 | #2 \rangle}
\newcommand\ketbra[2]{|#1 \rangle \hspace{-1pt} \langle #2 |}

\usepackage{amsbsy}
\usepackage[normalem]{ulem}
\allowdisplaybreaks[1]

\begin{document}

%+Title
\title{Quantum Latin squares and unitary error bases}
\author{\begin{tabular}{c@{\hspace{20pt}}c}
Benjamin Musto & Jamie Vicary
\\
\texttt{benjamin.musto@cs.ox.ac.uk}
&
\texttt{jamie.vicary@cs.ox.ac.uk}
\end{tabular}
\\[20pt]
Department of Computer Science, University of Oxford}
\date{\today}
\maketitle
%-Title

%+Abstract
\begin{abstract}
In this paper we introduce \textit{quantum Latin squares}, combinatorial quantum objects which generalize classical Latin squares, and investigate their applications in quantum computer science. Our main results are on applications to \textit{unitary error bases} (UEBs), basic structures in quantum information which lie at the heart of procedures such as  teleportation, dense coding and error correction. We present a new method for constructing a UEB from a quantum Latin square equipped with extra data. Developing construction techniques for UEBs has been a major activity in quantum computation, with three primary methods proposed: \textit{shift-and-multiply}, \textit{Hadamard}, and \textit{algebraic}. We show that our new approach simultaneously generalizes the shift-and-multiply and Hadamard methods. Furthermore, we explicitly construct a UEB using our technique which we prove cannot be obtained from any of these existing methods.
\end{abstract}

\section{Introduction}

We begin with the definition of a quantum Latin square.

\begin{definition}
A \textit{quantum Latin square of order $n$} is an  $n$-by-$n$ array of elements of the Hilbert space $\C^n$, such that every row and every column is an orthonormal basis. 
\end{definition}

\begin{example}
\label{ex:qls}
Here is a quantum Latin square given in terms of the computational basis elements \mbox{$\{\ket 0, \ket 1, \ket 2, \ket 3\} \subset \C ^4$:}
\begin{equation*}
\grid{\ket{0} & \ket{1} & \ket{2} & \ket{3}
\\\hline
\frac{1}{\sqrt{2}}(\ket{1}-\ket{2})
& \frac{1}{\sqrt{5}}(i\ket{0}+2\ket{3})
& \frac{1}{\sqrt{5}}(2\ket{0}+i\ket{3})
& \frac{1}{\sqrt{2}}(\ket{1}+\ket{2})
\\\hline
\frac{1}{\sqrt{2}}(\ket{1}+\ket{2})
& \frac{1}{\sqrt{5}}(2\ket{0}+i\ket{3})
& \frac{1}{\sqrt{5}}(i\ket{0}+2\ket{3})
& \frac{1}{\sqrt{2}}(\ket{1}-\ket{2})
\\\hline
\ket{3} & \ket{2} & \ket{1} & \ket{0}}
\end{equation*}
\end{example}

\noindent
It can readily be checked that along each row, and along each column, the elements form an orthonormal basis for $\C^4$.
We can compare this to the classical notion of Latin square~\cite{latinsquare}.
\begin{definition}
\label{def:ls}
A \textit{classical Latin square of order $n$} is an $n$-by-$n$ array of integers in the range $\{0, \ldots, n-1\}$, such that every row and column contains each number exactly once.
\end{definition}

\noindent
By interpreting a number $k \in \{0, \dots, n-1\}$ as a computational basis element $\ket k \in \C^n$, we can turn an array of numbers into an array of Hilbert space elements:
\newcommand\kp[1]{\phantom{|} #1 \phantom{\rangle}}
\begin{equation}
\grid{
\kp 3 & \kp 1 & \kp 0 & \kp 2
\\\hline
1 & 0 & 2 & 3
\\\hline
2 & 3 & 1 & 0
\\\hline
0 & 2 & 3 & 1
}
\qquad\leadsto\qquad
\grid{
\ket 3 & \ket 1 & \ket 0 & \ket 2
\\\hline
\ket 1 & \ket 0 & \ket 2 & \ket 3
\\\hline
\ket 2 & \ket 3 & \ket 1 & \ket 0
\\\hline
\ket 0 & \ket 2 & \ket 3 & \ket 1
}
\end{equation}
It is easy to see that the original array of numbers is a classical Latin square if and only if the corresponding grid of Hilbert space elements is a quantum Latin square. However, as Example~\ref{ex:qls} makes clear, not every quantum Latin square is of this form.

Our main results are on the construction of \textit{unitary error bases} (UEBs)~\cite{klapp}, also known as unitary operator bases. These are basic structures in quantum information which play a central role in quantum teleportation~\cite{teleportation}, dense coding~\cite{dense} and error correction~\cite{shor}. Since UEBs are hard to find, and given their wide applicability, construction techniques for UEBs have been widely studied~\cite{knill,klapp,werner2001all,ghosh}. In this paper, we propose a new method for construction of UEBs:
\begin{itemize}
\item Quantum shift-and-multiply method (\textbf{QSM}). Requires a quantum Latin square and a family of Hadamard matrices. (See Definition~\ref{def:qsm}.)
\end{itemize}
We compare this to the other methods that have been proposed in the literature:
\begin{itemize}
\item Shift-and-multiply method (\textbf{SM}). Requires a classical Latin square and a family of Hadamard matrices. (See Definition~\ref{def:sm}.)
\item Hadamard method (\textbf{HAD}). Requires a pair of mutually-unbiased bases. (See Definition~\ref{def:hadueb}.)
\item Algebraic method (\textbf{ALG}). Requires a finite group equipped with a projective representation, satisfying certain properties. (See Definition~\ref{def:algueb}.)
\end{itemize}
Our theorems concern the relationships between these constructions. In Theorems~\ref{thm:smqsm} and~\ref{thm:MUBQLS}, we prove that \textbf{QSM} contains \textbf{SM} and \textbf{HAD} as special cases. We also use \textbf{QSM} to construct a concrete unitary error basis $\mathcal M$ (Example~\ref{ex:qlsueb}), and prove that it is not equivalent to one arising from \textbf{SM}, \textbf{HAD} or \textbf{ALG} (Corollaries~\ref{cor:MnotSM}, \ref{cor:Mnothad} and~\ref{cor:Mnotnice} respectively.)

The relationships between these constructions, up to a standard notion of equivalence of UEBs (see Definition~\ref{def:equivUEB}), are indicated by the following Venn diagram:
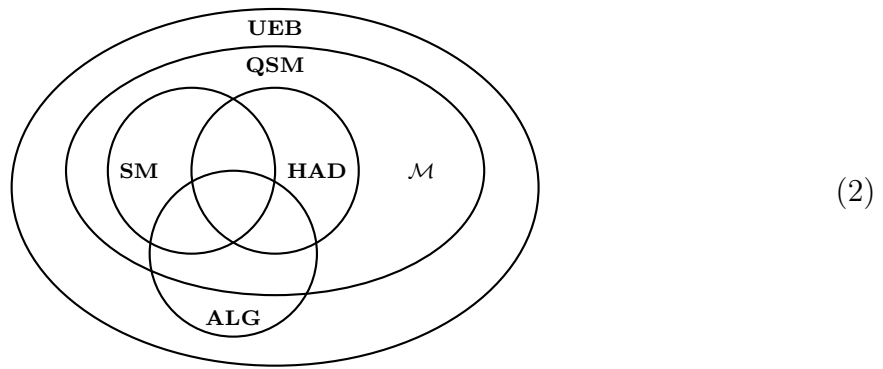
\begin{equation}
\def\s{1.1}
\begin{aligned}
\begin{tikzpicture}[thick, scale=\s, font=\scriptsize, xscale=1]
\node [ellipse, draw, minimum width=\s*2cm, minimum height=\s*2cm] at (0,0) {};
\node [ellipse, draw, minimum width=\s*2cm, minimum height=\s*2cm] at (1,0) {};
\node [ellipse, draw, minimum width=\s*2cm, minimum height=\s*2cm] at (0.5,-1) {};
\node [ellipse, draw, minimum width=\s*5cm, minimum height=\s*3cm] at (1,0) {};
\node [ellipse, draw, minimum width=\s*6.3cm, minimum height=\s*4.3cm] at (1,-0.20) {};
\node [below] at (1,1.5) {\textbf{QSM}};
\node [below] at (1,1.95) {\textbf{UEB}};
\node [above] at (0.5,-2) {\textbf{ALG}};
\node [right] at (-1,0) {\textbf{SM}};
\node [left] at (2,0) {\textbf{HAD}};
\node at (2.75,0) {$\mathcal M$};
\end{tikzpicture}
\end{aligned}
\end{equation}
Our work strongly extends previous results, in an area that has not seen progress since 2003. But there is much still to be settled: in particular, we do not know whether  \textbf{ALG} is a subset of \textbf{QSM}, or whether \textbf{QSM} equals \textbf{UEB}.

Categorical quantum mechanics is a research programme in which powerful techniques of monoidal category theory are used to understand quantum computational phenomena~\cite{bob-book, abramskycoecke2004, surveycategoricalquantummechanics}, using a graphical notation which can make the high-level structure of computations easier to understand. The main results of this paper were originally developed using this approach~(see also \cite{mustothesis}), although we have chosen to present them here in a conventional way. We feel this is a good advert for the power of categorical quantum mechanics; certainly, we could not have developed our results without using these techniques.

There are interesting connections between Hadamard matrices, unitary error bases and quantum Latin squares. In Section~\ref{sec:HadQLS} we show that a quantum Latin square can be constructed from any Hadamard matrix. Hadamard matrices are mathematically equivalent to the data for a pair of \textit{mutually unbiased bases}~\cite{mubs}, the study and classification of which is a major activity in quantum computer science~\cite{mafu, test, bart, speng}. It has also been shown that in some cases a family of mutually unbiased bases can be extracted from a UEB~\cite{uebtomub}. So quantum Latin squares can be built from Hadamards, which can be built from UEBs, which can be built from quantum Latin squares; an interesting tapestry of results for which we currently lack a good intuition.

\vspace{5pt}\noindent
\textbf{Acknowledgements.}
The authors are grateful to Dominic Verdon for useful discussions, and to EPSRC for financial support.

\section{\mbox{Quantum Latin squares from Hadamard matrices}}
\label{sec:HadQLS}

In this section we introduce some basic properties of quantum Latin squares, show how to construct a quantum Latin square from a Hadamard matrix, and prove that our quantum Latin square of Example~\ref{ex:qls} is not equivalent to one arising in this way.

We begin by developing a precise notation for working with quantum Latin squares. Throughout, we assume we are working with a quantum Latin square of order $n$, and that indices $i,j,k,p,q$ range from $0$ to $n-1$.

\begin{definition}
\label{def:notation}
For a quantum Latin square $Q$, we define the following:
\begin{itemize}
\item $Q_i$ is the matrix whose columns are the entries of the $i$th row of $Q$;
\item $Q_{ij} \in \C^n$ is the Hilbert space element at the $i$th row and $j$th column of $Q$;
\item $Q_{ijk} := (Q_{ij})_k = \langle k| Q_{ij} \rangle\in \C$ is the coefficient of the basis vector $\ket k$.
\end{itemize}
\end{definition}

\noindent
For a matrix $M$, it is a standard notation to write $M_{ij}$ for the element at the $i$th row and $j$th column. Combining this with Definition~\ref{def:notation}, we have the following:
\begin{equation}
\label{eq:indexswap}
(Q_i ^\pdag ) _{jk} ^\pdag = Q_{ikj} ^\pdag
\end{equation}
Note that the order of the final two indices changes.

Given a collection of numbers $Q_{ijk} \in \C$, we can easily identify when they arise from a quantum Latin square. For a matrix $M$, we write $M ^*$ for the conjugate matrix, $M^\T$ for the transpose matrix, and $M ^\dag = (M^*)^\dag = (M^\dag)^*$ for the conjugate transpose matrix.
\begin{lemma}
\label{lem:qls12}
A family of numbers $Q_{ijk} \in \C$ arise from a quantum Latin square if and only if they satisfy the following properties for all $i,p,q$:
\begin{align}
\label{eq:qls1}
\textstyle \sum _{j} Q_{ipj} ^* Q_{iqj} ^\pdag &= \delta _{pq} ^\pdag
\text{, or equivalently the matrices $Q_i$ are unitary}
\\
\label{eq:qls2}
\textstyle \sum _j Q _{pij} ^* Q_{qij} ^\pdag &= \delta _{pq}
\end{align}
\end{lemma}

\begin{proof}
Equations~\eqref{eq:qls1} and~\eqref{eq:qls2} are exactly the condition that the rows and columns, respectively, of the quantum Latin square form orthonormal bases. Unitarity of $Q_i$ means precisely $(Q ^\dag _i \circ Q_i ^\pdag)_{pq} = \delta _{pq}$, which expands to $\sum _j (Q_i ^\dag)_{pj} ^\pdag (Q_i ^\pdag) _{jq} ^\pdag = \sum_j Q ^* _{ipj} Q _{iqj} ^\pdag = \delta _{pq}$. (Recall that for an operator $Q$ on a finite-dimensional Hilbert space, $Q \circ Q^\dag = \I_n$ if and only if $Q ^\dag \circ Q = \I_n$.)
\end{proof}

\noindent
The condition~\eqref{eq:qls2} equivalently says that the matrices formed by the \textit{columns} of the Latin square are unitary, but this is not a fact that we will need directly.

There are certain trivial ways to transform a quantum Latin square into a different quantum Latin square, which we use to define a notion of equivalence.
\begin{definition}
\label{def:qlsequiv}
Two quantum Latin squares are \textit{equivalent} when one can be obtained from the other by permuting rows and columns, multiplying rows and columns by unit complex numbers, and applying a fixed unitary to every element. Algebraically, quantum Latin squares $Q$ and $Q'$ are equivalent when there exists some unitary $U$, diagonal unitary $D$, permutation matrix $P$,  permutation $\phi$, and a family of unit complex numbers $c_j$, such that the following holds:
\begin{equation}
\label{eqn:qlsequiv}
Q'_j = c_j U \circ Q_{\phi(j)} \circ P \circ D
\end{equation}

\end{definition}

We now give the standard definition of a Hadamard matrix, as a square matrix with entries of absolute value 1 which is proportional to a unitary matrix. 
\begin{definition}[See \cite{hadamard}, Definition~2.1]
A \textit{Hadamard matrix of order n} is an $n$-by-$n$ matrix $H$ with the following properties for all $i,j$, which we write in both matrix and index form:
\begin{align}
\label{had1}
|H_{ij}| &=1
&
H_{ij} ^\pdag H_{ij} ^* &= 1
\\
\label{had2}
H \circ H^\dag &= n \,\mathbb{I}_n
& \textstyle \sum_p H_{ip} ^\pdag H^* _{jp} &= n \, \delta _{ij}
\\
\label{had3}
H ^\dag \circ H &= n \,\mathbb{I}_n
& \textstyle \sum_p H ^*_{pi} H _{pj} ^\pdag &= n \, \delta _{ij}
\end{align}
\end{definition} 

\ignore{We illustrate this with an example:
\begin{equation}
\diag\left(
\begin{pmatrix}
   1 & 1 & 1& 1 \\
   1 & i & -1 & -i \\
   1 & -1 & 1 & -1 \\
   1 & -i & -1 & i
\end{pmatrix}
,1 \right)
=
\begin{pmatrix}
   1&0&0&0
   \\
   0&i&0&0
   \\
   0&0&-1&0
   \\
   0&0&0&-i
\end{pmatrix}
\end{equation}}
\begin{definition}[See~\cite{werner2001all}, Section 4]
\label{def:equivhad}
Two Hadamard matrices are \textit{equivalent} when one can be obtained from the other by permuting rows and columns, and multiplying rows and columns by unit complex numbers. Algebraically, $H,H'$ are equivalent if there exist $P_1,P_2$ permutation matrices and $D_1,D_2$ unitary diagonal matrices such that:
\begin{equation}
H'=D_1 \circ P_1 \circ H \circ P_2 \circ D_2 
\end{equation}
\end{definition}

We now give the construction of a quantum Latin square from a Hadamard matrix.

\begin{definition}
For a square matrix $M$, let $\diag(M,i)$ be the diagonal matrix whose diagonal entries are given by the $i$th row of $M$:
\begin{equation}
\label{eq:diag}
\diag(M,i)_{jk} := \delta _{jk} M_{ij}
\end{equation}
\end{definition}

\begin{definition}
\label{def:hadqls}
For a Hadamard matrix $H$ of order $n$, its \textit{associated quantum Latin square} $Q_H$ of order $n$ is defined as follows:
\begin{equation}
\label{eq:hadqls}
(Q_H)_j:= \textstyle \frac{1}{n}H \circ \diag(H,j)^\dag \circ \textstyle H^{\dag}
\end{equation}
We will refer to a quantum Latin square constructed in this way as a \textit{Hadamard quantum Latin square}.
\end{definition}
\begin{theorem}
\label{lem:MUBQLS}
The associated quantum Latin square construction is correct.
\end{theorem}
\begin{proof}
To establish property~\eqref{eq:qls1}, we note that $(Q_H)_j$ is the composite of three unitary matrices, and is therefore unitary.\ignore{To establish property~\eqref{eq:qls2}, we note that the unitary matrix $\frac{1}{\sqrt{n}}H$ diagonalizes $(Q_H)_j$ for all $j$:
\begin{align}
\textstyle
\big(\frac{1}{\sqrt{n}}H^{\dag} \big)
\circ (Q_H)_j \circ
\big( \frac{1}{\sqrt{n}}H \big)
\super {\eqref{eq:hadqls}} =
\textstyle \frac{1}{n^2}H^{\dag} \circ H \circ \diag(H,j)^\dag \circ H^\dag \circ H \super{\eqref{had2}}
= \diag (H,j)^\dag
\end{align}
So the columns of $\frac{1}{\sqrt{n}} H$ are the eigenvectors of $(Q_H)_j$ and the entries of $\diag(H,j)^{\dag}$ are its eigenvalues. Let $\ket{a_m}$ be the $m$th column of $\frac{1}{\sqrt{n}} H$ and let $\ket{b_j}^T=(b_{0j},...,b_{n-1,j})$ be the $j$th row of $\frac{1}{\sqrt{n}} H$. Then  the spectral decomposition of $(Q_H)_j$ is $(Q_H)_j =\sum_m b_{mj}^* \ket{a_m}\bra{a_m}$. Thus:
\begin{align}
\label{eq:1}(Q_H)_{pk} &=\sum_n b_{np}^* \ket{a_n} \braket{a_n|k}
\\ 
\label{eq:2}\text{and } 
(Q_H)_{qk} &=\sum_m b_{mq}^* \ket{a_m} \braket{a_m|k}\end{align}

And finally:
\begin{align*}
\sum _j Q _{pij} ^* Q_{qij} ^\pdag =\braket{(Q_H)_{pk,}(Q_H)_{qk}}\super{~\eqref{eq:1}\eqref{eq:2}}=\sum_{mn} b_{np}^{\pdag}b_{mq}^* \braket{k|a_n}\braket{a_n|a_m} \braket{a_m|k} 
\\ =
\sum_{m} b_{mp}^{\pdag}b_{mq}^* \braket{k|a_m} \braket{a_m|k}
=
\sum_{m} b_{mp} ^{\pdag}b_{mq}^* \braket{k|k}
=
\sum_{m} b_{mp}b_{mq}^*=
\braket{b_q|b_p}
\super{\eqref{lem:Qiunitary}}=\delta_{pq}
\end{align*}
as required.}
To verify~\eqref{eq:qls2}, we write expression~\eqref{eq:hadqls} in index form:
\begin{align}
\nonumber
&(Q_H) _{qij}
\super {\eqref{eq:indexswap}} =
((Q_H) _q) _{ji}
\super {\eqref{eq:hadqls}} =
\textstyle
\frac 1 n \sum _{rs} H_{jr}^\pdag \diag(H,q) ^\dag _{rs} H^\dag _{si}
\\
\label{eq:qhindex}
\,\,\,
&\super {\eqref{eq:diag}} = \textstyle
\frac 1 n \sum _{rs} H_{jr}^\pdag H^* _{qr} \delta _{rs}^\pdag H^* _{is}
=
\textstyle \frac 1 n \sum_r H_{jr} ^\pdag H_{qr} ^* H_{ir} ^*
\intertext{We then perform the following calculation:}
\nonumber
&\textstyle \sum _j (Q_H) ^* _{pij} (Q_H) _{qij} ^\pdag
\super {\eqref{eq:qhindex}} =
\textstyle
\frac 1 {n^2} \sum_j \big( \sum_r H^* _{jr} H _{pr} ^\pdag H_{ir} ^\pdag \big) \big( \sum_s H_{js} ^\pdag H_{qs} ^* H_{is} ^* \big)
\\
\nonumber
&\,\,\,
=
\textstyle \frac 1 {n^2} \sum _{rs} \big( \sum_j H_{jr} ^* H_{js} ^\pdag \big) H_{pr} ^\pdag H_{ir} ^\pdag H_{qs} ^* H_{is} ^*
\super {\eqref{had3}} =
\frac 1 n \sum_{rs} \delta _{rs} ^\pdag H_{pr} ^\pdag H_{ir} ^\pdag H_{qs} ^* H_{is} ^*
\\
&\,\,\,= \textstyle \frac 1 n \sum_r H_{pr} ^\pdag H_{qr} ^*  H_{ir} ^\pdag H_{ir} ^*
\super {\eqref{had1}} =
\frac 1 n \sum_r H_{pr} ^\pdag H_{qr}^*
\super {\eqref{had2}} =
\delta _{pq} ^\pdag 
\end{align}
In the second equality here, the sum is being reorganized.
\end{proof}
We now establish a lemma which we will use to prove Lemma~\ref{lem:MnotHad} and later Proposition~\ref{equivhadUEB1}.
\begin{lemma}
\label{lem:diag}
Let $p$ be the permutation associated with the permutation matrix $P$ such that $P=\sum_k\ketbra{p(k)}{k}$ and $D$ be a diagonal unitary. Then the following equations hold:
\begin{align}
\label{eqn:diag1}
\diag(P \circ H,i)
&=
\diag(H,p(i))=\diag(H_{p(i),0},...,H_{p(i),n-1})
\\
\label{eqn:diag2}
\diag(H \circ P,i)
&=
\diag(H_{i,p(0)} , ... , H_{i,p(n-1)})
\\
\label{eqn:diag3}
\diag(D \circ H,i)
&=
D_{ii} \, \diag(H,i)
\\
\label{eqn:diag4}
\diag(H \circ D,i)
&=
D \, \circ \diag(H,i)=\diag(H,i)\circ D
\end{align}
\end{lemma}

\begin{proof}
Straightforward calculation.
\end{proof}
\begin{lemma}
\label{lem:MnotHad}
Equivalent Hadamards give rise to equivalent quantum Latin squares.
\end{lemma}
\begin{proof}
We will prove equivalence on a case-by-case basis. 
Suppose $H' = P\circ H$. Then we have the following, where we use the fact that $P ^{-1} = P^\dag = P^T$:
\begin{align}
(Q_{H'})_{j}
&\super{\eqref{eq:hadqls}}=
\textstyle\frac{1}{n}P \circ H \circ \diag(P \circ H,j)^{\dag} \circ H^{\dag}  \circ P^{-1}\nonumber
\\
& \super {\eqref{eqn:diag1}} =
\textstyle\frac{1}{n}P \circ H \circ \diag(H,p(j))^{\dag} \circ H^{\dag}  \circ P^{-1}\nonumber
\\
& \super{\eqref{eqn:qlsequiv}} \sim
\textstyle\frac{1}{n}H \circ \diag(H,j)^{\dag} \circ H^{\dag}\super{\eqref{eq:hadqls}}=(Q_H)_j \nonumber 
\intertext{We now consider $H' = H \circ P$:}
(Q_{H'})_{j}
&\super{\eqref{eq:hadqls}}= \textstyle\frac{1}{n}H \circ P \circ \diag(H \circ P,j)^{\dag} \circ P^{-1} \circ H^{\dag}\nonumber
\\
& \super {\ensuremath{\eqref{eqn:diag2}}} {=}
\textstyle\frac{1}{n}H \circ P \circ \diag(H_{j,p(0)},...,H_{j,p(n-1)})^{\dag} \circ P^{-1} \circ H^{\dag}\nonumber 
\\
& \super{\eqref{eq:DPPD}} = 
\textstyle\frac{1}{n}H \circ P \circ P^{-1} \circ \diag(H,j)^{\dag} \circ H^{\dag}\nonumber
\\
&= \textstyle\frac{1}{n}H \circ \diag(H,j)^{\dag} \circ H^{\dag}\super{\eqref{eq:hadqls}}=(Q_H)_j \nonumber
\intertext{Finally, suppose $H' = D_1 \circ H \circ D_2$, with $D_1 = \diag(c_1, \ldots, c_n)$, where $|c_i| = 1$. Then we calculate as follows:}
 (Q_{H'})_{j}
 &\super{\eqref{eq:hadqls}}=
\textstyle\frac{1}{n}D_1 \circ H \circ D_2 \circ \diag(D_1 \circ H \circ D_2,j)^{\dag} D_2^{\dag}\circ H^\dag \circ D_1^{\dag}\nonumber
\\
& \super {\eqref{eqn:diag3}} = 
\textstyle\frac{1}{n}D_1 \circ H \circ D_2 \circ c_j \diag(H\circ D_2,j)^{\dag} \circ D_2^{\dag} \circ H^\dag \circ D_1 \nonumber 
\\
& \super {\eqref{eqn:diag4}} = 
\textstyle\frac{1}{n}D_1 \circ H \circ D_2 \circ c_j \diag(H,j)^{\dag}D_2 \circ D_2^{\dag} \circ H^\dag \circ D_1 \nonumber
\\
& = \textstyle\frac{1}{n}D_1 \circ H \circ D_2 \circ c_j \diag(H,j)^{\dag}\circ H^\dag \circ D_1 \nonumber
\\
 & \super{\eqref{eqn:qlsequiv}} \sim \textstyle \frac{1}{\sqrt{n}}\diag(H,j)^{\dag} \circ H^\dag \nonumber
\\
& \super{\eqref{eqn:qlsequiv}} \sim \textstyle \frac{1}{n}H \circ\diag(H,j)^{\dag} \circ H^\dag \super{\eqref{eq:hadqls}}= (Q_H)_j \nonumber
\end{align}
This completes the proof.
\end{proof}

Finally, we prove that our example quantum Latin square does not arise in this way, even up to equivalence. This makes use of some results that we prove later in the paper.
\begin{proposition}
\label{prop:MnotHad}
The quantum Latin square given in Example~\ref{ex:qls} is not equivalent to a quantum Latin square constructed from a Hadamard.
\begin{proof}
Let $H_\alpha$ be the family of Hadamard matrices as defined in equation~\eqref{eq:4had}, let  $(Q_{H_\alpha})_j:= \textstyle \frac{1}{n}H_\alpha \circ \diag(H_\alpha,j)^\dag \circ \textstyle H_\alpha^{\dag}$ be the associated quantum Latin squares, and let $Q$ be the quantum Latin square of Example~\ref{ex:qls}. By Lemma~\ref{lem:MnotHad} and Proposition~\ref{equivhadUEB2}, any quantum Latin square arising from a Hadamard matrix in the manner of Definition~\ref{def:hadqls} is equivalent to $Q_{H_\alpha}$ for some value of $\alpha$. 

For a contradiction, suppose that $Q$ and $Q_{H_\alpha}$ are equivalent in the manner of Definition~\ref{def:qlsequiv}, for some fixed value of $\alpha$. So there exists some unitary matrix $U$, diagonal unitary matrix $D$, permutation matrix $P$, permutation $\phi$, and a family of unit complex numbers $c_j$, such that the following holds:
\begin{equation*}
(Q_{H_\alpha})_j = c_j U \circ Q_{\phi(j)} \circ P \circ D
\end{equation*}
Note that the composite $P \circ D$ is unitary; so the families of matrices $(Q_{H_\alpha})_j$ and $Q_j$, which are unitary by Lemma~\ref{lem:qls12}, are equivalent families in the sense of Definition~\ref{def:equivUEB}.

The family $(Q_{H _\alpha})_j$ are simultaneously monomializable, by the matrix $Y$ defined in equation~\eqref{Y}. (This follows from Theorem~\ref{thm:d4mubmonomial}, in which we show that the members of $\F_\alpha$, which include the $(Q_{H_\alpha})_j$ as a subset, are simultaneously monomializable.)
So all together, the family of matrices $Q_j$ contains the identity, and is equivalent in the sense of Definition~\ref{def:equivUEB} to a monomial family. So by Proposition~\ref{lem:UEBmonomial}, the family $Q_j$ is simultaneously monomializable, and thus by Proposition~\ref{prop:simmon}, their 12th powers must all commute. But as established in the proof of Theorem~\ref{thm:Mnotmon}, the 12th powers of $Q_1 = \mathcal M_{01}$ and $Q_2 = \mathcal M _{02}$ do not commute. This gives us our contradiction.\end{proof}
\end{proposition}

\section{Unitary error bases from    quantum Latin squares}

In this section we define unitary error bases, and present our new \textit{quantum shift-and-multiply} construction, which produces a unitary error basis from a quantum Latin square equipped with a family of Hadamard matrices.
We then introduce an example UEB $\M$, which will play an important role in later sections where we show that it cannot arise from the shift-and-multiply, Hadamard or algebraic methods, even up to equivalence. 

We begin with the definition of unitary error basis. As remarked in the introduction, these structures play a central role in quantum computation.
\begin{definition}[See~\cite{klapp}, Section 1]
For a Hilbert space $H$ of dimension $n$, a \textit{unitary error basis} (or \textit{unitary operator basis}) is a family of $n^2$ unitary matrices $U_{ij}: H \to H$ which form an orthogonal basis:
\begin{equation}
\Tr (U^{\dag}_{ij} \circ U ^{\vphantom{\dag}}_{i'j'}) = \delta _{ii'} \delta _{jj'} n
\end{equation}
\end{definition}

\noindent
There is a standard notion of equivalence of unitary error bases, which we recall here.

\begin{definition}[See~\cite{klapp}, Section~2]
\label{def:equivUEB}
Two families of unitary matrices $\mathcal{A}$, $\mathcal B$ are \textit{equivalent} if there are unitary matrices $U$ and $V$, such that for any element $A \in \mathcal A$, there is an element $B \in \mathcal B$ and a unit complex number $c$ such that the following holds:
\begin{equation}
\label{eq:UEBequiv}
B = c \, U \circ A \circ V
\end{equation}
\end{definition}

\noindent
The following technical lemma will be useful later.

\begin{lemma}
\label{lem:composezerodiagonal}
Let $D$ be a diagonal matrix, and $A$ be a square matrix which is zero along the main diagonal, such that $D$ and $A$ are composable. Then $D \circ A$ is zero along the main diagonal.
\end{lemma}
\begin{proof}
We perform the following calculation of the diagonal elements of $D \circ A$:
\begin{equation}
\textstyle
(D \circ A) _{ii} = \sum _k D _{ik} A _{ki} = \sum_k \delta _{ik} D_{ii} A_{ki} = D_{ii} A_{ii} = 0
\end{equation}
 Here we apply the definition of matrix composition, the diagonal property of $D$, the properties of the sum, and the hypothesis that $A$ is zero along the main diagonal.
\end{proof}

We now define the main construction of focus in this paper. This construction is similar to Werner's shift-and-multiply method~\cite{werner2001all}, the difference being that ours is in terms of \textit{quantum} Latin squares. As usual, we take all indices in the range $0$ to $n-1$.
\begin{definition}[Quantum shift-and-multiply method]
\label{def:qsm}
Let $Q$ be a quantum Latin square of order $n$, and $H_j$ be a family of $n$ Hadamard matrices  of order $n$.  Then the associated \textit{quantum shift-and-multiply basis} has the following elements:
\begin{equation}
\label{eq:defqls}
S_{ij} := Q_j \circ \diag(H_j,i)
\end{equation}
\end{definition}
\noindent
In words, the $(i,j)$ entry of the quantum shift-and-multiply basis is the matrix given by the $j$th row of the quantum Latin square, composed with the diagonal matrix formed from the $i$th row of the $j$th Hadamard matrix.

We illustrate this with an example. This example will play a central role, as we will show in the remainder of the paper that it cannot be obtained, even up to equivalence, by any of the existing methods of unitary error basis construction.
\begin{example}
\label{ex:qlsueb}
The quantum shift-and-multiply basis $\mathcal M$ is constructed from the quantum Latin square of Example~\ref{ex:qls}, and from the following family of Hadamard matrices:
\begin{equation}
\label{eqn:Mhad}
H_0 = H_1 = H_2 = H_3 = %F =
\begin{pmatrix}
   1 & 1 & 1& 1 \\
   1 & i & -1 & -i \\
   1 & -1 & 1 & -1 \\
   1 & -i & -1 & i
\end{pmatrix}
\end{equation}
The resulting family of 16 matrices is listed in Appendix~\ref{sec:M}.
\end{example}

We now show that quantum shift-and-multiply bases are unitary error bases. This has similarities with Werner's original proof~\cite{werner2001all} for standard shift-and-multiply bases (see Section~\ref{sec:shiftandmultiply}), but our use of quantum Latin squares requires nontrivial extra ideas.
\begin{theorem}
\label{thm:main}
Quantum shift-and-multiply bases are unitary error bases.
\end{theorem}

\begin{proof}
First, we note that the elements $S_{ij} = Q_j \circ \diag(H_j,i)$ are unitary, since they are composites of unitary matrices: the matrix $Q_j$ is the $j$th row of a quantum Latin square, and hence unitary by Lemma~\ref{lem:qls12}; and $\diag(H_j,i)$ is a diagonal matrix with unit complex numbers along the diagonal, and hence unitary.

We must establish the following trace property:
\begin{equation}
\Tr (S_{ij} ^\dag \circ S ^\pdag _{i'j'}) = n \, \delta _{ii'} \delta _{jj'}
\end{equation}
We first consider the case that $j= j'$ and $i=i'$. By unitarity of $S_{ij}$ we have $S_{ij}^\pdag \circ S_{i'j'} ^\dag = \I_n$, with $\Tr(\I_n) = n$, and so the condition follows.

Next we consider the case that $j=j'$ and $i \neq i'$. We perform the following calculation:
\begin{align*}
 \Tr (S^{\dag}_{ij} \circ S ^\pdag_{i'j})
 &\super{\eqref{eq:defqls}}= \Tr \big( \diag (H_j, i)^{\dag} \circ Q^{\dag}_j \circ Q_{j}^\pdag \circ \diag(H_{j}, i') \big)
\\
&\super{\eqref{eq:qls1}}=
\Tr \big( \diag(H_j,i)^{\dag} \circ \diag(H_{j}, i') \big)
 \end{align*}
  The final expression is equal to the inner product of rows $i$ and $i'$ of the Hadamard $H_j$. Since distinct rows of a Hadamard are orthogonal, the result is zero as required.

It remains to consider the case that $j \neq j'$. We use the cyclic property of the trace to rearrange our trace expression:
\begin{align}
\nonumber
\Tr(S_{ij} ^\dag \circ S ^\pdag _{i'j'}) &\super{\eqref{eq:defqls}}= \Tr \big( \diag (H_j, i) ^\dag \circ Q_j ^\dag \circ Q_{j'}^\pdag \circ \diag(H_{j'}, i') \big)
\\
\label{eq:rearrangedtrace}
&= \Tr \big( \diag (H_{j'}, i') \circ \diag (H_j, i) ^\dag \circ Q_j ^\dag \circ Q_{j'}^\pdag \big)
\end{align}
Inside the trace there is the composite $\diag(H_{j'},i') \circ \diag(H_j, i)^\dag$, which is diagonal. There is also $Q_j ^\dag \circ Q_{j'}^\pdag$, which by the following argument is zero along the diagonal:
\begin{align}
\textstyle
&\textstyle
(Q_{j} ^\dag \circ Q_{j'} ^\pdag) _{kk} ^\pdag = \sum _{l} (Q ^\dag _{j}) _{kl} ^\pdag (Q_{j'} ^\pdag) _{lk} ^\pdag = \sum _{l} (Q^* _{j}) _{lk} ^\pdag (Q_{j'} ^\pdag) _{lk} ^\pdag
\super {\eqref{eq:indexswap}} =
\sum _{l} Q_{jkl} ^{*\pdag} Q ^\pdag_{j'kl}
\super {\eqref{eq:qls2}} =
\delta _{jj'} ^\pdag = 0
\end{align}
Hence by Lemma~\ref{lem:composezerodiagonal}, expression~\eqref{eq:rearrangedtrace} is zero as required.
\end{proof} 

\section{Shift-and-multiply method}
\label{sec:shiftandmultiply}

The shift-and-multiply method of Werner~\cite{werner2001all}, which was a direct inspiration for our own results,  can straightforwardly be seen as a special case of our quantum shift-and-multiply method. Our focus in this section is the proof that the unitary error basis $\M$ of Example~\ref{ex:qlsueb} is not equivalent to a shift-and-multiply basis, and thus that the shift-and-multiply bases are \textit{strictly} contained within the quantum shift-and-multiply bases.
\begin{definition}
\label{def:sm}
A \textit{shift-and-multiply basis} is a quantum shift-and-multiply basis where the quantum Latin square is a classical Latin square.
\end{definition}
\begin{theorem}
\label{thm:smqsm}
Every shift-and-multiply basis is a quantum shift-and-multiply basis.
\end{theorem}
\begin{proof}
Follows immediately from Definitions~\ref{def:ls} and~\ref{def:sm}.
\end{proof}

Monomial matrices will be crucial to our proof strategy.
\begin{definition}
A \textit{monomial matrix} is a square matrix with exactly one nonzero entry in each row and each column. Equivalently, it is any matrix $A$ which can be expressed as $A = D_A \circ P_A$, where $D_A$ is a diagonal matrix and $P_A$ is a permutation matrix.
\end{definition}
\begin{lemma}
\label{lem:DPPD}
Let $p$ be a permutation, $P=\sum_k\ketbra{p(k)}{k}$ be the corresponding permutation matrix, and $D=\sum_kd_k\ketbra{k}{k}$ and $D' = \sum _k d _{p(k)} \ketbra{k}{k}$ be diagonal matrices. Then the following holds:
\begin{equation}
\label{eq:DPPD}
D \circ P = P \circ D',
\end{equation}
\end{lemma}
\begin{proof}
We perform the following calculation:
\begin{align*}
&
\textstyle
D \circ P = P \circ P^{\dag} \circ D \circ P = P \circ \big( \sum_{ijk}
\ketbra{i}{p(i)}d_j\ket{j}
\braket{j}{p(k)}
\bra{k} \big)
\\
&
\textstyle
=P \circ \big( \sum_{ik} d_{p(k)}\ket{i}\braket{p(i)}{p(k)}\bra{k} \big) = P \circ \big( \sum_i d_{p(i)}\ketbra{i}{i} \big)=P \circ D'
\end{align*}
This completes the proof.
\end{proof}

\begin{lemma}
\label{lem:monomialproperties}
The set of monomial matrices is closed under composition, taking inverses, taking adjoints, and multiplication by nonzero complex scalars.
\end{lemma}
\begin{proof}
Straightforward.
\end{proof}

\begin{definition}
A square matrix $A$ is \textit{monomializable} if there exists a unitary matrix $U$ such that $U \circ A \circ U^{\dag}$ is monomial. 
\end{definition}
\begin{definition}
A family of square matrices $A_1,...,A_n$ are \textit{simultaneously monomializable} if they are all monomializable by the same unitary matrix $U$.
\end{definition}
\ignore{\begin{definition}
A \textit{monomial basis} is a UEB in which every matrix is monomial.  
\end{definition}}

\ignore{
\begin{proof}
I THINK WE CAN OMIT THIS, IT'S QUITE AN INTUITIVE RESULT. For \ref{enum:mon1}, 
For monomial matrices $A,B$ with monomial decompositions $D_AP_A$ and $D_BP_B$ respectively: 

 $\exists D'_B,D''$ diagonal matrices and $P_{AB}=P_AP_B$  such that: 
\begin{equation}
AB=D_AP_AD_BP_B=D_AD'_BP_AP_B=D''P_{AB}
\end{equation}
The first equality is by Definition 8, the second holds due to Lemma 4, and the final equality is due to the fact that both the diagonal and permutation matrices are closed under matrix composition.

For property \ref{enum:mon2}: Let $A$ be a monomial matrix with $A=DP$ for $D$ diagonal and $P$ a permutation. Then $A^{-1}=P^{-1}D^{-1}=D'P^{-1}$ by Lemma 4 and since the diagonal and the permutation matrices are closed under inverse.
\end{proof}
}
We establish the following propositions, the first of which is adapted and generalized to suit our purposes from the literature.
\begin{proposition}[See \cite{klapp}, final part of the proof of Theorem~3]
\label{lem:UEBmonomial}
If a family $\mathcal S$ of unitary matrices containing the identity is equivalent (in the sense of Definition~\ref{def:equivUEB}) to a family of monomial matrices, then the members of $\mathcal S$ are simultaneously monomializable.
\end{proposition}
\begin{proof}
Let $\mathcal S = \{S_i\}$ be a family of unitary matrices with $S_{0} = \I_n$. Suppose $S_{i}$ is equivalent to some monomial family $\mathcal T= \{T_i\}$ with $T_i = c_{i} U S_i V$, such that each $c_{i}$ is a complex number of norm 1, and $U,V$ are unitary matrices. We then perform the following calculation:
\begin{equation}
\frac{c_0}{c_j}T _j ^\pdag T _0 ^\dag= \frac {c_0} {c_j} (c_j U S_j V)(c_0 U S_0 V) ^\dag = c_0^{\vphantom{*}} c_0^* U S_j V V^\dag \I_n U^\dag = U S_j U^\dag
\end{equation} 
The left hand side  is monomial by Lemma~\ref{lem:monomialproperties}, and hence $U$ simultaneously monomializes~$S_i$.
\end{proof}
\ignore{
\begin{proof}
GOOD CANDIDATE FOR OMITTING, SINCE IT IS NOT OUR PROOF ANYWAY? Let our UEB be $S_{ij}$ $|$ $0 \leq i,j<d$ , and let $S_{00}$ be the identity. Suppose $S_{ij}$ is equivalent to some monomial basis $T_{ij}=c_{ij}AS_{ij}B$ such that $c_{ij} \in \mathbb{C}$ and $A,B$ are unitary matrices. Then
$T_{00}:=c_{00}AS_{00}B=c_{00}AB$ is monomial. 

$\forall i,j \, \, \, 0 \leq i,j<d:$
\begin{equation}
T_{ij}=c_{ij}AS_{ij}B=c_{ij}AS_{ij}A^{\dag}AB=\frac{c_{ij}}{c_{00}}AS_{ij}A^{\dag}T_{00} \\ \Rightarrow \frac{c_{00}}{c_{ij}}T_{ij}T_{00}^{\dag}=AS_{ij}A^{\dag}
\end{equation} 
The left hand side of this final equation is monomial by Lemma 5. So given a UEB\ with identity $S_{ij}$, equivalent to a monomial basis,  $\exists$ a unitary $A$ such that $AS_{ij}A^{\dag}$ for all $i,j$ are monomial, ie the whole basis is simultaneously monomializable.
\end{proof}
}
\begin{proposition}
\label{prop:simmon}
Let $A,B$ be square matrices of size $n$, and let $\mu_n$ be the lowest common multiple of $\{1,2,...,n\}$. If $A$ and $B$ are simultaneously monomializable, then $A^{\mu_n}$ and $B^{\mu_n}$ commute.
\end{proposition}

\begin{proof}
Suppose $A,B$ are simultaneously monomializable,  with $\mu_n$ defined as above. Then there exists a unitary matrix $U$ such that $UAU^{\dag}=D_AP_A$ and $UBU^{\dag}=D_BP_B$ where $D_A,D_B$ are diagonal matrices and $P_A,P_B$ are permutation matrices. Note that $A = U^\dag D_A P_A U$, so we have the following:
\begin{equation}
A^{\mu_n}=U^{\dag} (D_AP_A)^{\mu_n} U = U ^\dag  \widetilde{D}_A^\pdag P_A ^{\mu_n} U
\end{equation} 
Here $\widetilde{D}_A$ is some diagonal matrix, and the last equality is obtained by repeated application of Lemma~\ref{lem:DPPD} and the fact that diagonal matrices are closed under composition. Since $P_A$ is a permutation matrix of dimension $n$  it has order $k$, where $k$ is the lowest common multiple of the lengths of the permutation's cycles. Each cycle has length $\in \{1,2,...,n\}$. Thus $k$ divides $\mu_n$, and so $P_A^{\mu_n} = \mathbb I_n$. So $A^{\mu_n} = U^{\dag} \widetilde{D}_A U$, and by the same argument, $B ^{\mu_n} = U ^\dag \widetilde{D}_B U$ for some diagonal matrix $\widetilde{D}_B$. We then demonstrate that $A ^{\mu_n}$ and $B ^{\mu_n}$ commute:
\begin{equation}
\nonumber
A ^{\mu_n} B ^{\mu_n} = U ^\dag \widetilde{D}_A U U ^\dag \widetilde{D}_B U = U ^\dag \widetilde{D}_A \widetilde{D}_B U = U ^\dag \widetilde{D}_B \widetilde{D}_A U = U ^\dag \widetilde{D}_B U U ^\dag \widetilde{D}_A U = B ^{\mu_n} A ^{\mu_n}
\end{equation}
The central equality here holds because diagonal matrices commute.
\end{proof}

We are now ready to prove the necessary properties of our example basis.
\begin{theorem}
\label{thm:Mnotmon}
The basis $\mathcal M$ of Example~\ref{ex:qlsueb} is not equivalent to a monomial basis.
\end{theorem}

\begin{proof}
For a contradiction, suppose that $\mathcal M$ is equivalent to a monomial basis. Note that $\mathcal M$ contains the identity matrix, so by Proposition~\ref{lem:UEBmonomial} the elements of the UEB are simultaneously monomializable. The least common multiple of $\{1,2,3,4\}$ is $\mu_4 = 12$; thus by Proposition~\ref{prop:simmon} the 12th powers of the elements of $\mathcal M$ will commute. To exhibit the contradiction, we compute the following commutator: 
\begin{equation}
\label{eqn:commutator}
(\M_{01})^{12}(\M_{02})^{12}-(\M_{02})^{12}(\M_{01})^{12}=
\frac{12168}{15625}
\begin{pmatrix}
-i & 0 & 0 & 2 \\
0 & 0 & 0 & 0 \\
0 & 0 & 0 & 0 \\
-2 & 0 & 0 &  i  \\
\end{pmatrix} 
\neq 0
\end{equation}
This completes the proof.
\end{proof}

\begin{proposition}
\label{prop:smmonomial}
Shift-and-multiply bases are monomial bases.
\end{proposition}
\begin{proof}
Recall from Definition~\ref{def:sm} of a shift-and-multiply basis that each matrix is the product of a diagonal matrix with the permutation matrix given by a row of a classical Latin square. By definition, the result is a monomial matrix.
\end{proof}

\begin{corollary}
\label{cor:MnotSM}
The basis $\mathcal M$ of Example~\ref{ex:qlsueb} is not equivalent to a shift-and-multiply basis.
\end{corollary}
\begin{proof}
Immediate from Theorem~\ref{thm:Mnotmon} and Proposition~\ref{prop:smmonomial}.
\end{proof}

\section{Hadamard method}

In this section we study the \textit{Hadamard method}, a direct construction of a unitary error basis from a Hadamard matrix. While this is certainly known, we cannot find a clear description of it in full generality, although a special case is worked out in detail in~\cite{bobross}. The main results of this section are Theorem~\ref{thm:MUBQLS}, where we show that the quantum shift-and-multiply method contains the Hadamard method as a special case, and Corollary~\ref{cor:Mnothad}, in which we show that this containment is proper.\begin{definition}[Hadamard method; folklore]
\label{def:hadueb}
For a Hadamard matrix $H$ of order $n$, its associated \emph{Hadamard basis} $\{ (U_H) _{ij} \}$ is defined as follows:
\begin{equation}
\label{eqn:MUBUEB}
(U_H)_{ij}=\textstyle\frac{1}{n}H \circ \diag(H,j)^{\dag} \circ H^{\dag} \circ \diag(H^T,i)
\end{equation} 
\end{definition}
\begin{theorem}
\label{thm:MUBQLS}
A Hadamard basis is a quantum shift-and-multiply basis.
\end{theorem}
\begin{proof}
By Definition~\ref{def:hadqls} and Theorem~\ref{lem:MUBQLS} we have $(U_H)_{ij}=(Q_{H})_j \circ \diag(H^T,i)$. Since the transpose of a Hadamard is also a Hadamard, the result follows.
\end{proof}
\begin{corollary}
A Hadamard basis is a unitary error basis.
\end{corollary}
\begin{proof}
Follows from Theorems~\ref{thm:main} and~\ref{thm:MUBQLS}.
\end{proof}

\begin{proposition}
\label{equivhadUEB1}
If two Hadamard matrices are equivalent by Definition~\ref{def:equivhad}, then their associated unitary error bases are equivalent by Definition~\ref{def:equivUEB}.
\end{proposition}
\begin{proof}
We will once again prove equivalence on a case-by-case basis. 
Again suppose $H' = P\circ H$. Then we have the following:
\begin{align*}
(U_{H'})_{ij}&\super {\eqref{eqn:MUBUEB}}=\textstyle\frac{1}{n}P \circ H \circ \diag(P \circ H,j)^{\dag} \circ H^{\dag}  \circ P^{-1}\circ \diag(H^T\circ P^{-1},i)
\intertext{Again using the fact that $P$ is real and unitary  so, $P^{-1}=P^{\dag}=P^T$. We continue:}
(U_{H'}) _{ij}
& \super {{\eqref{eqn:diag1}}{\eqref{eqn:diag2}}}  = \,\,
\textstyle\frac{1}{n}P \circ H \circ \diag(H,p(j))^{\dag} \circ H^{\dag} \circ P^{-1}\circ \diag(a_{p(0),i},...,a_{p(n-1),i})
\\
& \super{\eqref{eq:DPPD}} = 
\, \, \textstyle\frac{1}{n}P \circ H \circ \diag(H,p(j))^{\dag} \circ H^{\dag}\circ \diag(a_{p^{-1}p(0),i},...,a_{p^{-1}p(n-1),i})\circ P^{-1}
\\
& \super{\eqref{eq:diag}} = \,\,
\textstyle\frac{1}{n}P \circ H \circ \diag(H,p(j))^{\dag} \circ H^{\dag}\circ \diag(H^T,i)\circ P^{-1}
\\
& \super{\eqref{eq:UEBequiv}} \sim \,\,
\textstyle\frac{1}{n}H \circ \diag(H,p(j))^{\dag} \circ H^{\dag}\circ \diag(H^T,i)
\\
& \super {\eqref{eqn:MUBUEB}} =
\, \,(U_H)_{i,p(j)}
\intertext{The case that $H' = H \circ P$ is similar. 
Now suppose $H' = D \circ H$, with $D = \diag(c_1, \ldots, c_n)$, where $|c_i| = 1$. Then we calculate as follows:}
(U_{H'})_{ij}
 & \super {\eqref{eqn:MUBUEB}} =\,\,
\textstyle\frac{1}{n}D \circ H \circ \diag(D \circ H,j)^{\dag} \circ H^\dag \circ D^{\dag} \circ \diag(H^{T} \circ D ^T,i)
\\
& \super{\eqref{eqn:diag3}} = \,\, 
\textstyle\frac{1}{n}D \circ H \circ c_j \diag(H,j)^{\dag} \circ H^{\dag}\circ \diag(H^{T} \circ D,i) \circ D ^\dag \quad
\\
& \super{\eqref{eqn:diag4}} = \,\,
\textstyle\frac{c_j}{n} D \circ  H \circ \diag(H,j)^{\dag} \circ H^{\dag}\circ \diag(H^{T},i) \circ D \circ D^\dag
\\
& \super{\eqref{eq:UEBequiv}}  \sim \,\,
\textstyle\frac{1}{n}H \circ \diag(H,j)^{\dag} \circ H^{\dag}\circ\diag(H^{T},i)
\\
& \super {\eqref{eqn:MUBUEB}} = \,\,
(U_H)_{ij}
\end{align*}
The case $H' = H \circ D$ is similar.
\end{proof}
\begin{proposition}[See~\cite{had4}, Theorem 1]
\label{equivhadUEB2}
All Hadamard matrices on $\C^4$ are equivalent to one of the following Fourier matrices, parameterised by $\alpha \in [0,\frac{\pi}{2}]$:
\begin{equation}
\label{eq:4had}
H_{\alpha} :=
\begin{pmatrix}
1 & 1 & 1 &1 \\
1 & 1 & -1 & -1 \\
1 & -1 &  e^{i \alpha} &  -e^{i \alpha} \\
1 & -1 &  -e^{i \alpha} & e^{i\alpha} 
\end{pmatrix}
\end{equation}
\end{proposition}

\begin{theorem}
\label{thm:d4mubmonomial}
Every unitary error basis for $\C^4$ arising from the Hadamard method is equivalent to a monomial basis.
\end{theorem}
\begin{proof}
Write $\mathcal F _\alpha$ for the unitary error basis arising from $H_\alpha$ by the Hadamard method, for some fixed $\alpha \in[0,\frac{\pi}{2}]$. By Propositions~\ref{equivhadUEB1} and~\ref{equivhadUEB2} all unitary error bases arising from Hadamards in dimension $4$ are equivalent to $\mathcal{F} _\alpha$, for some value of $\alpha$. But the following unitary matrix simultaneously monomializes $\mathcal{F} _\alpha$, for all values of $\alpha$:
\begin{equation}
\label{Y}
Y:=\frac{1}{\sqrt{2}}
\begin{pmatrix}
0 & 0 &-1 & 1 \\
-1 & 1 & 0 & 0 \\
0 & 0 & 1 & 1 \\
1 & 1 & 0 & 0
\end{pmatrix}
\end{equation}
The basis $\mathcal{F_\alpha'}=\{Y \circ F_{ij} \circ Y^{\dag}|F_{ij} \in \mathcal{F} \}$ is listed in Section~\ref{sec:Fprime}, and is monomial and equivalent to $\mathcal F_\alpha$. This completes the proof.
\end{proof}

\begin{corollary}
\label{cor:Mnothad}
The basis $\mathcal M$ of Example~\ref{ex:qlsueb} is not equivalent to a Hadamard basis.
\end{corollary}
\begin{proof}
Follows from Theorems~\ref{thm:Mnotmon} and~\ref{thm:d4mubmonomial}.
\end{proof}

\section{Algebraic method}

Another technique for constructing UEBs is the \textit{algebraic method}, due to Knill~\cite{knill}. UEBs obtained using this technique are called \textit{nice error bases}. The main result in this section is Corollary~\ref{cor:Mnotnice}, that the basis $\mathcal M$ of Example~\ref{ex:qlsueb} is not equivalent to a nice error basis. Throughout this section, we use `$\propto$' to denote equality up to multiplication by a unit complex number.

Recall that for a finite group $G$, an \textit{$n$-dimensional unitary projective representation} is a function $\rho: G \to U(n)$, valued in the group of $n$-by-$n$ unitary matrices, and for any $g,g' \in G$ a complex number $\omega _{g,g'} \in \C$ with unit norm, such that we have $\rho(gg')=\omega_{g,g'}\rho(g) \rho(g')$ and $\rho(1)=\mathbb{I}_n$ where $1$ is the group identity. We therefore have the following:
\begin{equation}
\label{eq:ghprop}
\rho(g) \rho(g')
\propto \rho(gg')
\text{ for all $g,g' \in G$}
\end{equation}
The following result will also be useful.
\begin{lemma}
Given a unitary projective representation $\rho$ of a group $G$, the following holds:
\begin{equation}
\label{eq:rhoginverse}
\textstyle
\rho (g) ^\dag \propto \rho(g ^\inv )
\mathrm{\text{\rm{} for all $g \in G$}}
\end{equation}
\end{lemma}
\begin{proof}As follows:
\mbox{$\!\rho(g) ^\dag
=
\rho(g)^\dag \rho(1) = \rho(g)^\dag \rho(gg^\inv)
\super {\eqref{eq:ghprop}} \propto
\rho(g) ^\dag \rho(g) \rho(g^\inv) = \rho(g ^{-1})$.\hspace{25pt}}\mbox{\hspace{-70pt}}
\end{proof}

We now give the definition of a nice error basis, and show that a nice error basis is a unitary error basis.
\begin{definition}[Nice error basis. See~\cite{knill}, Section 2]
\label{def:algueb}
Let $G$ be a finite group of order $n^2$, and let $\rho$ be an $n$-dimensional unitary projective representation of $G$, such that for all $g \in G$ not equal to the identity, we have the following:
\begin{equation}
\label{eq:tr0}
\Tr(\rho(g)) = 0
\end{equation}
Then a \textit{nice error basis} $\mathcal R _{G,\rho} := \{ \rho(g) \,|\, g \in G \}$ is the image of $\rho$.
\end{definition}

\begin{lemma}[See \cite{klapp}, Lemma 3]
A nice error basis is a unitary error basis.
\end{lemma}
\ignore{
\begin{proof}
NOT OUR PROOF, CANDIDATE FOR OMISSION. It is clear that the elements of a nice error basis are unitary. The trace condition is satisfied, since for all $g \neq g'$ we have the following:
\begin{equation}
\Tr(\rho(g') \rho(g) ^\dag)
\super {\eqref{eq:rhoginverse}} \propto
\Tr(\rho(g') \rho (g ^\inv))
\super {\eqref{eq:ghprop}} \propto
\Tr(\rho(g' g^\inv))
\super {\eqref{eq:tr0}} =
0
\end{equation}
The final equality follows since $g \neq g'$ exactly when $g' g ^\inv \neq 1$. Since $G$ has order $n^2$, the unitary error basis conditions are satisfied.
\end{proof}
}

We now prove a key proposition, which we will use to establish that our example basis $\M$ of Example~\ref{ex:qlsueb} is not equivalent to a nice error basis.
\def\I{\mathbb{I}}
\def\S{\mathcal{S}}
\begin{proposition}
\label{thm:closedunderadjoints}
Let $\mathcal S$ be a unitary error basis containing the identity matrix $\I_n$, such that $\S$ is equivalent to a nice error basis. Then up to multiplication by a unit complex number, $\S$ is closed under taking adjoints.
\end{proposition}
\begin{proof}
Let $\mathcal R _{G,\rho}$ be a nice error basis, and let $\mathcal S = \{ c_g U \rho(g) V \,| \,g \in G\}$ be an equivalent unitary error basis, with elements $\S_g := c_g U \rho(g) V$. Since by hypothesis $\I_n \in \mathcal S$, there is some $h \in G$ with $\S_h = c_h U \rho(h) V = \I_n$. In particular, writing `$\propto$' to indicate equality up to multiplication by a unit complex number, we have the following:
\begin{align}
\label{eq:Ipropto}
\I_n &\propto U \rho(h) V
\\
\label{eq:Sgprop}
\S_g &\propto U \rho(g) V
\text{ for all $g \in G$}
\end{align}
We now perform the following calculation, for any $g \in G$:
\begin{align}
\nonumber
& (\S_g) ^\dag
\super {\eqref{eq:Sgprop}} \propto
V ^\dag \rho(g) ^\dag U^\dag = \I_n V^\dag \rho(g) ^\dag U^\dag \I_n
\super {\eqref{eq:Ipropto}} \propto
U \rho(h) V V^\dag \rho(g) ^\dag U^\dag U \rho(h) V
\\
\nonumber
&
=
U \rho(h) \rho(g) ^\dag \rho (h) V
\super {\eqref{eq:rhoginverse}} \propto
U \rho(h) \rho(g ^\inv) \rho(h) V
\super {\eqref{eq:ghprop}} \propto
U \rho(h g^\inv h) V
\super {\eqref{eq:Sgprop}} \propto
\S_{h g ^\inv h}
\end{align}
So $\S$ is closed under adjoints, up to multiplication by a unit complex number.
\end{proof}
\begin{corollary}
\label{cor:Mnotnice}
The basis $\mathcal M$ of Example~\ref{ex:qlsueb} is not equivalent to a nice error
basis.
\end{corollary}
\begin{proof}
By inspection of the elements of $\mathcal M$, as listed in Section~\ref{sec:M}. For a contradiction, let us assume that $\M$ is equivalent to a nice error basis. Note that $\mathcal M$ contains the identity matrix; then by Proposition~\ref{thm:closedunderadjoints}, it must be closed under taking adjoints, up to a unit complex number. But this is clearly false: for example, the second element of the first row of $\M_{01}$ has absolute value $\frac 1 {\sqrt 5}$, but no member of $\mathcal M$ has an element with the same absolute value in the second element of the first column.
\end{proof}

\bibliographystyle{plainurl}
\bibliography{UEBPaper}

\newpage
\appendix
\section{Lists of unitary error bases}
Here we list the unitary error bases that we make use of in the main text.
\def\minus{\text{-}}

%\small
%\footnotesize
%\scriptsize

\subsection{The unitary error basis $\M$}
\label{sec:M}
Here we list the unitary error basis $\M$ defined in Example~\ref{ex:qlsueb}.
\def\matrixgap{\hspace{5pt}}
\def\beginmatrix{\left(\!\!\begin{array}{c@{\matrixgap}c@{\matrixgap}c@{\matrixgap}c}}
\def\endmatrix{\end{array}\!\!\right)}
\def\vs{\vphantom{\frac 2 {\sqrt{5}}}}
\[
\scriptsize{ 
\begin{array}{llll}
\M_{00} =
\beginmatrix
1 & 0 & 0 & 0 \\
 0 & 1 & 0 & 0 \\
 0 & 0 & 1 & 0 \\
 0 & 0 & 0 & 1 \\
\endmatrix
&
\M_{01} =
\beginmatrix
 0 & \frac{i}{\sqrt{5}} & \frac{2}{\sqrt{5}} & 0 \\
 \frac{1}{\sqrt{2}} & 0 & 0 & \frac{1}{\sqrt{2}} \\
 \minus\frac{1}{\sqrt{2}} & 0 & 0 & \frac{1}{\sqrt{2}} \\
 0 & \frac{2}{\sqrt{5}} & \frac{i}{\sqrt{5}} & 0 \\
\endmatrix
&
\M_{02} =
\beginmatrix
 0 & \frac{2}{\sqrt{5}} & \frac{i}{\sqrt{5}} & 0 \\
 \frac{1}{\sqrt{2}} & 0 & 0 & \frac{1}{\sqrt{2}} \\
 \frac{1}{\sqrt{2}} & 0 & 0 & \minus\frac{1}{\sqrt{2}} \\
 0 & \frac{i}{\sqrt{5}} & \frac{2}{\sqrt{5}} & 0 \\
\endmatrix
&
\M_{03} =
\beginmatrix
 0\vs & 0 & 0 & 1 \\
 0\vs & 0 & 1 & 0 \\
 0\vs & 1 & 0 & 0 \\
 1\vs & 0 & 0 & 0 \\
\endmatrix
\\
\M_{10} =
\beginmatrix
 1 & 0 & 0 & 0 \\
 0 & i & 0 & 0 \\
 0 & 0 & \minus 1 & 0 \\
 0 & 0 & 0 & \minus i \\
\endmatrix
&
\M_{11} =
\beginmatrix
 0 & \minus \frac{1}{\sqrt{5}} & \minus\frac{2}{\sqrt{5}} & 0 \\
 \frac{1}{\sqrt{2}} & 0 & 0 & \minus\frac{i}{\sqrt{2}} \\
 \minus\frac{1}{\sqrt{2}} & 0 & 0 & \minus\frac{i}{\sqrt{2}} \\
 0 & \frac{2 i}{\sqrt{5}} & \minus\frac{i}{\sqrt{5}} & 0 \\
\endmatrix
 &
\M_{12} =
\beginmatrix
 0 & \frac{2 i}{\sqrt{5}} & \minus\frac{i}{\sqrt{5}} & 0 \\
 \frac{1}{\sqrt{2}} & 0 & 0 & \minus\frac{i}{\sqrt{2}} \\
 \frac{1}{\sqrt{2}} & 0 & 0 & \frac{i}{\sqrt{2}} \\
 0 & \minus\frac{1}{\sqrt{5}} & \minus\frac{2}{\sqrt{5}} & 0 \\
\endmatrix
&
\M_{13} =
\beginmatrix
 0\vs & 0 & 0 & \minus i \\
 0\vs & 0 & \minus1 & 0 \\
 0\vs & i & 0 & 0 \\
 1\vs & 0 & 0 & 0 \\
\endmatrix
\\
\M_{20} =
\beginmatrix
 1 & 0 & 0 & 0 \\
 0 & \minus1 & 0 & 0 \\
 0 & 0 & 1 & 0 \\
 0 & 0 & 0 & \minus1 \\
\endmatrix
&
\M_{21} =
\beginmatrix
 0 & \minus \frac{i}{\sqrt{5}} & \frac{2}{\sqrt{5}} & 0 \\
 \frac{1}{\sqrt{2}} & 0 & 0 & \minus\frac{1}{\sqrt{2}} \\
 \minus\frac{1}{\sqrt{2}} & 0 & 0 & \minus\frac{1}{\sqrt{2}} \\
 0 & \minus \frac{2}{\sqrt{5}} & \frac{i}{\sqrt{5}} & 0 \\
\endmatrix
&
\M_{22} =
\beginmatrix
 0 & \minus\frac{2}{\sqrt{5}} & \frac{i}{\sqrt{5}} & 0 \\
 \frac{1}{\sqrt{2}} & 0 & 0 & \minus\frac{1}{\sqrt{2}} \\
 \frac{1}{\sqrt{2}} & 0 & 0 & \frac{1}{\sqrt{2}} \\
 0 & \minus\frac{i}{\sqrt{5}} & \frac{2}{\sqrt{5}} & 0 \\
\endmatrix
&
\M_{23} =
\beginmatrix
 0\vs & 0 & 0 & \minus1 \\
 0\vs & 0 & 1 & 0 \\
 0\vs & \minus 1 & 0 & 0 \\
 1\vs & 0 & 0 & 0 \\
\endmatrix
\\
\M_{30} =
\beginmatrix
 1 & 0 & 0 & 0 \\
 0 & \minus i & 0 & 0 \\
 0 & 0 & \minus1 & 0 \\
 0 & 0 & 0 & i \\
\endmatrix
&
\M_{31} =
\beginmatrix
 0 & \frac{1}{\sqrt{5}} & \minus\frac{2}{\sqrt{5}} & 0 \\
 \frac{1}{\sqrt{2}} & 0 & 0 & \frac{i}{\sqrt{2}} \\
 \minus\frac{1}{\sqrt{2}} & 0 & 0 & \frac{i}{\sqrt{2}} \\
 0 & \minus\frac{2 i}{\sqrt{5}} & \minus\frac{i}{\sqrt{5}} & 0 \\
\endmatrix
&
\M_{32} =
\beginmatrix
 0 & \minus\frac{2 i}{\sqrt{5}} & \minus\frac{i}{\sqrt{5}} & 0 \\
 \frac{1}{\sqrt{2}} & 0 & 0 & \frac{i}{\sqrt{2}} \\
 \frac{1}{\sqrt{2}} & 0 & 0 & \minus\frac{i}{\sqrt{2}} \\
 0 & \frac{1}{\sqrt{5}} & \minus\frac{2}{\sqrt{5}} & 0 \\
\endmatrix
&
\M_{33} =
\beginmatrix
 0\vs & 0 & 0 & i \\
 0\vs & 0 & \minus1 & 0 \\
 0\vs & \minus i & 0 & 0 \\
 1\vs & 0 & 0 & 0 \\
\endmatrix
\end{array}
}
\]

\subsection{The unitary error basis $\F'$}
\label{sec:Fprime}
Here we list the unitary error basis $\mathcal{F'}$ defined in the proof of Theorem~\ref{thm:d4mubmonomial}. 
\[
\scriptsize\begin{array}{llll}
\F'_{00} = 
\beginmatrix
  1 & 0 & 0 & 0 \\
 0 &  1 & 0 & 0 \\
 0 & 0 & 1 & 0 \\
 0 & 0 & 0 & 1 \\
\endmatrix
&
\F'_{01}=
\beginmatrix
 \minus 1 & 0 & 0 & 0 \\
 0 & \minus 1 & 0 & 0 \\
 0 & 0 & 1 & 0 \\
 0 & 0 & 0 & 1 \\
\endmatrix
&
\F'_{02}=
\beginmatrix
 0 & 1 & 0 & 0 \\
 e^{\minus 2 i a} & 0 & 0 & 0 \\
 0 & 0 & 0 & 1 \\
 0 & 0 & 1 & 0 \\
\endmatrix
&
\F'_{03}=
\beginmatrix
 0 & \minus 1 & 0 & 0 \\
 \minus e^{\minus 2 i a} & 0 & 0 & 0 \\
 0 & 0 & 0 & 1 \\
 0 & 0 & 1 & 0 \\
\endmatrix
\\
\F'_{10}=
\beginmatrix
 \minus 1 & 0 & 0 & 0 \\
 0 & 1 & 0 & 0 \\
 0 & 0 & \minus 1 & 0 \\
 0 & 0 & 0 & 1 \\
\endmatrix
& \F'_{11}=
\beginmatrix
 1 & 0 & 0 & 0 \\
 0 & \minus 1 & 0 & 0 \\
 0 & 0 & \minus 1 & 0 \\
 0 & 0 & 0 & 1 \\
\endmatrix
& \F'_{12}=
\beginmatrix
 0 & 1 & 0 & 0 \\
 \minus e^{\minus 2 i a} & 0 & 0 & 0 \\
 0 & 0 & 0 & 1 \\
 0 & 0 & \minus 1 & 0 \\
\endmatrix
& \F'_{13}=
\beginmatrix
 0 & \minus 1 & 0 & 0 \\
 e^{\minus 2 i a} & 0 & 0 & 0 \\
 0 & 0 & 0 & 1 \\
 0 & 0 & \minus 1 & 0 \\
\endmatrix
\\ \F'_{20}=
\beginmatrix
 0 & 0 & \minus e^{i a} & 0 \\
 0 & 0 & 0 & \minus 1 \\
 \minus e^{i a} & 0 & 0 & 0 \\
 0 & \minus 1 & 0 & 0 \\
\endmatrix
& \F'_{21}=
\beginmatrix
 0 & 0 & e^{i a} & 0 \\
 0 & 0 & 0 & 1 \\
 \minus e^{i a} & 0 & 0 & 0 \\
 0 & \minus 1 & 0 & 0 \\
\endmatrix
& \F'_{22}=
\beginmatrix
 0 & 0 & 0 & \minus 1 \\
 0 & 0 & \minus e^{\minus i a} & 0 \\
 0 & \minus 1 & 0 & 0 \\
 \minus e^{i a} & 0 & 0 & 0 \\
\endmatrix
& \F'_{23}=
\beginmatrix
 0 & 0 & 0 & 1 \\
 0 & 0 & e^{\minus i a} & 0 \\
 0 & \minus 1 & 0 & 0 \\
 \minus e^{i a} & 0 & 0 & 0 \\
\endmatrix
\\ \F'_{30}=
\beginmatrix
 0 & 0 & e^{i a} & 0 \\
 0 & 0 & 0 & \minus 1 \\
 e^{i a} & 0 & 0 & 0 \\
 0 & \minus 1 & 0 & 0 \\
\endmatrix
& \F'_{31}=
\beginmatrix
 0 & 0 & \minus e^{i a} & 0 \\
 0 & 0 & 0 & 1 \\
 e^{i a} & 0 & 0 & 0 \\
 0 & \minus 1 & 0 & 0 \\
\endmatrix
& \F'_{32}=
\beginmatrix
 0 & 0 & 0 & \minus 1 \\
 0 & 0 & e^{\minus i a} & 0 \\
 0 & \minus 1 & 0 & 0 \\
 e^{i a} & 0 & 0 & 0 \\
\endmatrix
&
\F'_{33}=
\beginmatrix
 0 & 0 & 0 & 1 \\
 0 & 0 & \minus e^{\minus i a} & 0 \\
 0 & \minus 1 & 0 & 0 \\
 e^{i a} & 0 & 0 & 0 \\
\endmatrix
\end{array}
\]

\end{document}